\documentclass[journal,romanappendices]{IEEEtran}
\usepackage{mathrsfs}
\usepackage{amsfonts}
\usepackage{amssymb}
\usepackage{dsfont}
\usepackage[dvips]{graphicx}
\usepackage{subfigure}
\usepackage{amsmath}

\usepackage{cite}

\usepackage{ctable}
\usepackage{multirow}

\newtheorem{theorem}{Theorem}
\newtheorem{lemma}{Lemma}
\newtheorem{corollary}{Corollary}
\newtheorem{definition}{Definition}

\begin{document}
\title{Bit Allocation Law for Multi-Antenna Channel Feedback Quantization: Single-User Case}
\author{Behrouz Khoshnevis and Wei Yu \thanks{The material in this paper has been presented
in part at the 25th Biennial Symposium on Communications, Kingston, Canada, May 2010.} \thanks{The authors are with the Edward S. Rogers Sr. Department of Electrical and Computer Engineering, University of Toronto, 10 King's College Road, Toronto, Ontario, Canada M5S 3G4 (email: bkhoshnevis@comm.utoronto.ca; weiyu@comm.utoronto.ca).}}
\maketitle

\newcommand{\qbar}{\bar{q}}
\newcommand{\gbar}{\bar{\gamma}}
\newcommand{\ds}{\mathds}
\newcommand{\md}{\textmd}
\newcommand{\be}{\begin{equation}}
\newcommand{\ee}{\end{equation}}
\newcommand{\ci}{\circ}
\newcommand{\eq}{\eqref}
\newcommand{\defined}{\stackrel{\md{def}}{=}}
\newcommand{\csi}{_{_{\md{CSI}}}}
\newcommand{\mbf}{\mathbf}
\newcommand{\mbb}{\mathbb}
\newcommand{\mc}{\mathcal}
\newcommand{\bh}{\mbf{h}}

\newcommand{\stupe}{\mbf{S}(\mbf{H})}
\newcommand{\bv}{\mbf{v}}

\newcommand{\D}{\mc{D}}
\newcommand{\M}{\mc{M}}
\newcommand{\NM}{\dot{N}}
\newcommand{\ND}{\ddot{N}}
\newcommand{\BM}{\dot{B}}
\newcommand{\BD}{\ddot{B}}
\newcommand{\NkM}{\dot{N}_k}
\newcommand{\NkD}{\ddot{N}_k}
\newcommand{\mcy}{\mathcal{\mathbf{y}}}
\newcommand{\bu}{\mbf{u}}
\newcommand{\bw}{\mbf{w}}

\newcommand{\BkD}{\ddot{B}_k}
\newcommand{\BkM}{\dot{B}_k}

\newcommand{\BbarM}{{\dot{{B}}_{\md{ave}}}}
\newcommand{\BbarD}{{\ddot{{B}}_{\md{ave}}}}

\newcommand{\sh}{\mc{S}(\mbf{h})}

\newcommand{\du}{\mathrm{d}}

\begin{abstract}
This paper studies the design and optimization of a limited feedback single-user system with multiple-antenna transmitter and single-antenna receiver. The design problem is cast in form of the minimizing the average transmission power at the base station subject to the user's outage probability constraint. The optimization is over the user's channel quantization codebook and the transmission power control function at the base station. Our approach is based on fixing the outage scenarios in advance and transforming the design problem into a robust system design problem. We start by showing that uniformly quantizing the channel magnitude in dB scale is asymptotically optimal, regardless of the magnitude distribution function. We derive the optimal uniform (in dB) channel magnitude codebook and combine it with a spatially uniform channel direction codebook to arrive at a product channel quantization codebook. We then optimize such a product structure in the asymptotic regime of $B\rightarrow \infty$, where $B$ is the total number of quantization feedback bits. The paper shows that for channels in the real space, the asymptotically optimal number of direction quantization bits should be ${(M{-}1)}/{2}$ times the number of magnitude quantization bits, where $M$ is the number of base station antennas. We also show that the performance of the designed system approaches the performance of the perfect channel state information system as $2^{-\frac{2B}{M+1}}$. For complex channels, the number of magnitude and direction quantization bits are related by a factor of $(M{-}1)$ and the system performance scales as $2^{-\frac{B}{M}}$ as $B\rightarrow\infty$.

\end{abstract}
\begin{keywords} Beamforming, bit allocation, channel quantization, limited feedback, multiple antennas, outage probability, power control.
\end{keywords}

\section{Introduction}\label{S_Yek}
It is well established that the use of multiple antennas at base station can considerably improve the performance of wireless communication links. The realization of these improvements often requires channel state information (CSI) at the base station. In time division duplex (TDD) systems with uplink and downlink channel reciprocity, the base station can acquire this information via training on the uplink channels. In frequency division duplex (FDD) systems, however, the remote terminal needs to quantize and feedback the channel information to the base station. Such systems are generally referred to as \emph{limited feedback} systems in the literature.

The availability of CSI at the base station can significantly improve the performance of the limited feedback single-user and multiuser systems. In a single-user system, limited feedback of CSI provides the base station with power gain. For the multiuser multiplexing systems, the availability of CSI is even more important. In such systems, the base station needs the users' channel state information to distinguish the users spatially and to perform rate and/or power control accordingly.

\subsection{Related Work}
\subsubsection{Single-User Systems}
Assuming perfect CSI, the capacity of single-user channels with single antenna can be achieved by variable-rate coding at the base station, where the transmission rate is continuously adapted to the channel magnitude \cite{6thesis,7thesis}. The practical implementation of variable-rate coding is achieved by adaptive modulation and coding (AMC), where the received signal-to-noise ratio (SNR) is compared with a certain set of thresholds and, based on this comparison, an appropriate modulation and coding index is chosen for the current channel realization \cite{wifi,wimax,LTE}. From the perspective of a limited feedback system design, these SNR thresholds are in essence the SNR quantization thresholds and effectively form a quantization codebook for the received SNR or equivalently for the channel magnitude. For an optimal system design, one therefore needs to optimize this codebook so that the average transmission rate is maximized for a given target decoding error probability at the receiver side \cite{6thesis,7thesis}.

Optimal design of limited feedback systems becomes more complicated when the stations are equipped with multiple antennas. The capacity of a limited feedback single-user channel with multiple antennas and temporally i.i.d. fading is shown to be achieved by sharing a codebook of transmission covariance matrices between the transmitter and receiver, where the receiver chooses the best transmit covariance matrix for the current channel realization and sends the corresponding index back to the transmitter \cite{4thesis,lau,14thesis}. Due to the complexities of the covariance codebook design, most of the work in this field focuses on either channel magnitude quantization \cite{kim,Bhash,KhoshSab} or channel direction quantization \cite{narula,heath,Mukk,xia,roh}, but not both. One can however easily show that for a limited-feedback system to approach the performance of the perfect CSI system, both channel direction and magnitude information are required \cite{globecom,queens}. It is therefore essential to study the design of magnitude and direction codebooks jointly and derive the optimal split of feedback bits between the two codebooks. The problem of joint codebook design and optimization has not been well addressed and formulated in the literature and it is the main subject of this paper.

\subsubsection{Multiuser Systems}
Although this paper is mainly concerned with single-user systems, it is informative to provide a short review of the literature on multiuser limited feedback systems.

Single-antenna multiuser systems are widely addressed in the literature, e.g. \cite{paper10,paper11}. The problem in general is to find a channel magnitude feedback mechanism that realizes the multiuser diversity gain, i.e. the scaling of the sum rate with the number of users in a large network of users. The design processes in these papers are however mainly numerical and they do not provide insight into the magnitude quantization codebook structure. The multiple-antenna multiuser systems are also widely studied in the literature \cite{26thesis,27thesis,24thesis,28thesis,29thesis,31thesis}. The availability of multiple antennas at the base station permits for spatial multiplexing of the users' data and for serving multiple users simultaneously. The work in \cite{24thesis} specifically shows that, for realizing the multiplexing gain, the number of feedback bits per user should scale logarithmically with SNR. In order to address the scheduling problem in large networks of users, a well-justified approach is to choose users with high channel magnitudes, low quantization errors, and almost orthogonal channel directions \cite{31thesis,27thesis}. In order to realize the multiuser diversity gain in such networks, \cite{31thesis} shows that one needs channel gain information (CGI) in addition to channel direction information (CDI); the CGI however is assumed to be unquantized in this work. The split of feedback bits between CGI and CDI quantization codebooks is studied in \cite{paper6}. Although optimal bit allocation laws and codebook structures are not provided by \cite{paper6} in a closed form, the authors are able to numerically show that as the number of users increase, more bits should be used for CGI quantization in order to benefit from multiuser diversity gain.

\subsection{Problem Formulation}
In this paper, we address the problem of channel magnitude and direction quantization codebook design and optimization for a multiple-input single-output (MISO) single-user channel.

We are mainly interested in the optimal split of feedback bits between the two codebooks and also deriving the scaling of the system performance with the number of feedback bits. The system design problem is formulated as minimizing the average transmission power at the base station subject to the user's target outage probability. The optimization is over the user's channel quantization codebooks and the power control function at the base station. The problem formulation used here is typically more appropriate for fixed-rate delay-sensitive applications, e.g. voice over internet protocol (VoIP) and fixed-rate video streaming applications. The power control strategy that results from this formulation tries to fix the transmission rate by compensating for the channel fading and in this sense resembles the power control mechanisms used in Wideband Code Division Multiple Access (WCDMA) system standards \cite{book_wcdma}. Alternatively, it is also possible to fix downlink transmission power and instead utilize adaptive modulation and coding to adapt the transmission rate to the user's channel quality as in Worldwide Interoperability for Microwave Access (WiMAX) and 3GPP Long Term Evolution (LTE) system standards \cite{book_lte}. Our problem formulation and results are not directly applicable to these latter models. The problem of optimal feedback bit split between channel magnitude and direction codebooks for such systems is not well studied in the literature and further analysis is required in this regard.

The problem formulation used in our work is somewhat similar to the formulations used by \cite{Gia2,Gia3,Gia4} for improving the power efficiency of limited feedback systems. The work in \cite{Gia2} specifically formulates the problem of minimizing the transmission power in a wireless sensor network subject to an average rate and average bit error rate constraint at the fusion center. The multiple single-antenna sensors in this work are assumed fully synchronized and therefore can be effectively modeled as a single multiple-antenna transmitter. Based on this model, the authors utilize adaptive modulation and coding and present a numerical approach to optimize the channel gain and direction quantization codebooks. The work however does not address the optimal split of the feedback bits between the codebooks and only numerically studies the scaling of the system performance with the number of feedback bits.

Our work is also related to the literature on vector quantization codebook design, although the deign objective used here is quite different from the classical design objectives, e.g. minimizing the mean squared error (MSE) distortion measure. It should also be noted that the channel magnitude and direction codebooks used in our analysis effectively form a product quantization codebook for the user's channel vector. This product structure, also known as shape-gain quantization \cite{gersho}, provides several practical advantages including faster quantization and lower storage requirement for the quantization codebooks. In addition, most practical systems already have power control modules that essentially act as channel magnitude quantizers; therefore, a product codebook structure is more easily adopted in such systems.

In order to study the problem of channel quantization codebook design, we first focus on the scalar channel magnitude quantization and characterize the optimal magnitude quantization codebook. It is shown that, as the codebook size $\NM$ increases, the optimal quantization levels form an approximately geometric sequence and hence become uniformly spaced in dB scale. Interestingly enough, the asymptotic optimality of uniform magnitude quantization (in dB scale) is not affected by the magnitude distribution function as long as certain regularity conditions are satisfied. This uniform in dB channel magnitude quantization codebook structure in fact provides a theoretical justification for using uniform (in dB) SNR quantization thresholds for power control in practical systems such as WCDMA and IEEE 802.11n system standards \cite{book_wcdma,wifi}.

For the optimal uniform channel magnitude quantization codebooks, we are able to show that the gradient of the average transmission power (the objective function) with respect to uniform (in dB) quantization levels should diminish as $\Theta({\NM}^{-3/2})$ as $\NM\rightarrow\infty$. We also derive a codebook that achieves such optimal scaling.

We next form a product channel vector quantization codebook comprising the optimal uniform (in dB) magnitude quantization codebook and a spatially uniform direction quantization codebook and study its performance in the asymptotic regime of $B\rightarrow\infty$, where $B$ is total number of feedback quantization bits. After deriving the optimal power control and beamforming function at the base station, we optimize the product quantization codebook structure such that the average transmission power is minimized. We show that, for channels in real space,
\begin{enumerate}
\item The optimal number of direction quantization bits is $\frac{M-1}{2}$ times the number of magnitude quantization bits, where $M$ is the number of base station antennas.
\item As $B$ increases, the system performance approaches the perfect CSI system as $2^{-\frac{2B}{M+1}}$.
\end{enumerate}

The reason that we focus on real space channels in this paper, is to compare the real-space magnitude-direction bit allocation laws with similar laws in the quantization theory literature, which are based on conventional mean squared error (MSE) distortion measure. For channels in the complex space, one can apply the exact same approach to show that the number of magnitude and direction quantization bits are related by a factor of $(M{-}1)$ and that the system performance scales as $2^{-\frac{B}{M}}$ in the asymptotic regime of $B{\rightarrow}\infty$ (see Section \ref{complex_channels}).

The remainder of this paper is organized as follows. Section \ref{S_Do} shows the asymptotic optimality of uniform (in dB) channel magnitude quantization. Section \ref{S_Se} presents the single-user system design problem in its general form and describes our approach in transforming it to a robust optimization problem. Section \ref{S_Char} describes the product codebook structure and derives the asymptotically optimal magnitude-direction bit allocation law. Section \ref{numerical_results} presents the numerical results and Section \ref{S_Panj} concludes the paper.

\emph{Notations:} All the computations in this paper are for the the real space. All logarithm functions throughput the paper are base 2. Also, the angle between any two unit-norm vectors $\bu$ and $\bv$ is defined as $\angle(\bu,\bv)=\arccos |\bu^T\bv|$ so that $0\leq \angle(\bu,\bv)\leq\pi/2$.

\section{Channel Magnitude Quantization}\label{S_Do}

This section studies the structure of the channel magnitude quantization codebook. The resulting codebook structure is used later in the paper to form a product vector channel quantization codebook. As mentioned earlier, we assume real space channels throughout the paper.

\subsection{SISO Limited Feedback System Optimization}

Consider a single-input single-output (SISO) channel $h\in\mathds{R}$. The input signal $s$ and the output signal $r$ are related as
\be\label{siso}
r=\sqrt{P(h)}hs+z,
\ee
where $\mathbb{E}[|s|^2]=1$, $P(h)$ is the transmission power, and $z\sim\mathcal{N}(0,1)$ is the Gaussian receiver noise.

Assume a perfect CSI system, where the base station is required to guarantee a fixed target SNR $\gamma_{_0}$ at the receiver. The optimal power control strategy is therefore channel inversion, where the base station sets the transmission power $P(h)$ according to \be\label{CSI-power} P(h)=\frac{\gamma_{_0}}{|h|^2}.\ee

Now consider a limited feedback system with perfect CSI at the receiver and a feedback link from the receiver back to the base station with a capacity of $B$ bits per fading block. We are interested in a system design that minimizes the average transmission power subject to a target outage probability $q$ at the user side. The outage probability is defined as the probability that the received SNR drops below a target SNR $\gamma_{_0}$.

To perform power control in a limited feedback system, the receiver quantizes the channel magnitude $|h|$ and sends the $B$-bit index of the quantized magnitude level back to the transmitter; the transmitter then sets the transmission power according to the quantized magnitude. The limited feedback system design is therefore is a two-fold problem: 1) optimizing the power control function given the magnitude quantization codebook; 2) optimizing the magnitude quantization codebook itself.

For the SISO system described above, the solution to the first problem is straightforward and can be expressed in a closed form as follows. Consider a channel magnitude quantization codebook \[\mbb{Y}=\{y^{(1)},y^{(2)},\cdots,y^{(\NM)}\},\] where $\NM$ is the codebook size\footnote{The dot notation $~\dot{(~~)}~$, throughput the paper, is used for all the variables associated with channel magnitude quantization. For the variables associated with channel direction quantization the double dot notation $~\ddot{(~~)}~$ is used.}. Here $y^{(n)}$'s are the magnitude quantization levels that are used for quantizing the magnitude variable $Y\defined |h|^2$. For a fixed codebook $\mbb{Y}$ of size $\NM$, we are interested in the power control function $P(h)$ that minimizes the average transmission power for a fixed target outage probability $q$. With this objective, the optimal strategy is to accumulate all the outage scenarios in the lowest magnitude interval, i.e. the leftmost interval in Fig. \ref{Fig1}, and set $P(h)$ in other intervals to the minimum possible value that prevents outage. This minimizes the average transmission power since preventing outage with a higher channel magnitude requires less transmission power.

\begin{figure}
\centering
\includegraphics[width=2in]{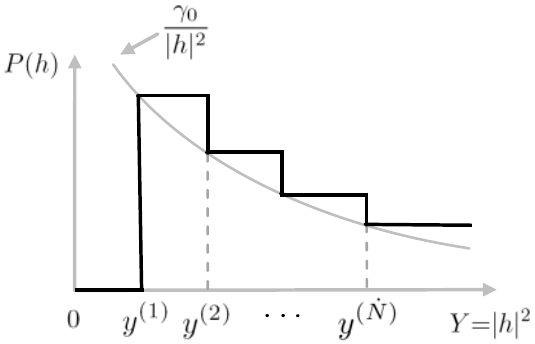}
\caption{The optimal transmission power $P(h)$ is the step-like function shown in bold line. The channel magnitude quantization codebook $\mbb{Y}=\{y^{(1)},y^{(2)},\cdots,y^{(\NM)}\}$ is fixed.}
\label{Fig1}
\end{figure}

Having the optimal transmission power in terms of the quantization codebook $\mbb{Y}$, we can now address the optimization of the quantization codebook itself. The objective is to find the codebook $\mbb{Y}$ of a given size $\NM$ that minimizes the average transmission power subject to a target outage probability $q$.

Define the \emph{quantized magnitude} $\tilde{Y}$ as the lower end of the quantization interval that includes $Y$:
\be\label{quant-mag} \tilde{Y}=y^{(n)} ~~~~ \md{if} ~~~~ y^{(n)}\leq Y <y^{(n+1)},\ee where $1{\leq} n{\leq} \NM$ and $y^{(\NM+1)}{\defined}\infty$.
The normalized average transmission power, according to Fig. \ref{Fig1}, is given by
\be\label{Pave}
\mc{P}(\mbb{Y})\defined\frac{1}{\gamma_{_{0}}}\mathbb{E}[P(h)]=\mbb{E}\left[\frac{1}{\tilde{Y}}\right]
=\sum_{n=1}^{\NM}{\frac{Q^{(n)}}{y^{(n)}}},
\ee
where $Q^{(n)}=F(y^{(n+1)})-F(y^{(n)})$ and $F(\cdot)$ is the cumulative distribution function (cdf) of $Y$.

According to Fig. \ref{Fig1}, the first quantization level $y^{(1)}$ is fixed by the outage probability $q$: \be \label{y1} y^{(1)}=F^{-1}(q),\ee where $F^{-1}(\cdot)$ is the inverse cdf of $Y$. In order to optimize the codebook $\mbb{Y}$, we take the derivative of $\mathcal{P}(\mbb{Y})$ with respect to $y^{(n)}$ for $2{\leq} n{\leq} \NM$ and set it equal to zero. This results in the following equation:
\be\label{der}
f(y^{(n)})\left[\frac{1}{y^{(n-1)}}-\frac{1}{y^{(n)}}\right]=\frac{Q^{(n)}}{(y^{(n)})^2},
\ee
where $f(\cdot)$ is the probability density function (pdf) of $Y$. By approximating
\be\label{approx}
Q^{(n)}\approx f(y^{(n)})(y^{(n+1)}-y^{(n)}),
\ee
for $2{\leq} n{\leq} \NM{-}1$, we achieve the following relation between the magnitude quantization levels:
\be\label{GeoResult}
\frac{y^{(n+1)}}{y^{(n)}}=\frac{y^{(n)}}{y^{(n-1)}},
\ee
i.e. the optimal quantization levels $y^{(n)}$, $1\leq n\leq \NM$, form an approximate geometric sequence and therefore are uniformly spaced in dB scale. This result, surprisingly enough, does not depend on the distribution of $Y=|h|^2$. The geometric ratio of such a sequence however should depend both on the codebook size and the distribution of $Y$ as it is shown next.

\subsection{Asymptotic Optimality of Uniform (in dB) Magnitude Quantization}

In order to rigorously show the efficiency of the geometric sequences for magnitude quantization, we proceed to bound the gradient of the objective function $\mathcal{P}$ in \eqref{Pave} with respect to such quantization levels. By showing that the norm of the gradient goes to zero as $\NM$ increases, we can then claim that such quantization codebooks are at least asymptotically optimal as $\NM\rightarrow\infty$.

Starting with some definitions, we define a geometric sequence $\mathbb{Y}^{(g)}(r)$ with the ratio parameter $r>1$ as follows:
\be\label{geoseq}
\mathbb{Y}^{(g)}(r)=\left\{\left.y^{(n)}=a{r^{n-1}}\right|1\leq n \leq \NM\right\},
\ee
where $a\defined y^{(1)}$ for notation convenience. We refer to these sequences as the uniform (in dB) magnitude quantization codebooks. In proving the asymptotic optimality of uniform (in dB) magnitude quantization codebooks, we will need the following regularity conditions on the channel magnitude distribution:

\begin{definition} \label{reg-cond} Magnitude distribution regularity conditions:
\begin{enumerate}
\item $f(y)$ is a positive differentiable function with bounded derivative over $y>0$.
\item $\eta=\lim_{y\rightarrow\infty}{{-f(y)}/{f'(y)}}$ exists and $\eta\neq 0$.
\item $\mbb{E}[Y]$ is bounded and therefore $\lim_{y\rightarrow\infty}yf(y)=0$.
\end{enumerate}
\end{definition}
The logic behind these assumptions becomes clear in the discussions to come.

In order to bound the gradient of the objective function $\mathcal{P}$ in \eqref{Pave}, one can use the Taylor expansion to bound the approximation error in \eqref{approx} and obtain the following bound on the gradient:
\begin{lemma} \label{L_Yek} For the uniform magnitude quantization codebook $\mathbb{Y}^{(g)}(r)$, we have
\be\label{grad1}
\left\|\nabla_{\mathbb{Y}^{(g)}(r)} \mathcal{P}\right\| < \frac{1}{2}\mu\sqrt{\NM} (r-1)^2+D,
\end{equation}
where $\mu$ is an upper bound of $f'(y)$ for all $y>0$, and
\begin{equation}
\label{Dterm}
D=\left|\frac{1}{y^2}\left((r-1)yf(y)-\int_{y}^{\infty}{f(t) \mathrm{d} t}\right)\right|_{y=y_{_{\NM}}=a{r^{\NM-1}}}
\end{equation}
\end{lemma}
\begin{proof}
See Appendix \ref{A_Yek}.
\end{proof}

As discussed earlier we want the gradient bound \eqref{grad1} to diminish to zero as the codebook size $\NM$ increases. To achieve this, the geometric ratio $r$ should be finely tuned with $\NM$. The idea is to set $r$ appropriately, so that both terms on the right-hand side of \eqref{grad1} go to zero as $\NM\rightarrow\infty$. Since $r>1$, we can express $r$ as a function of $\NM$ as \[r=1{+}\mc{L}(\NM),\] where $\mc{L}(\NM)$ is an arbitrary
positive function. Define \[\zeta(\NM) = -{\log \mathcal{L}(\NM)}/{\log \NM},\] so that \[r=1+\NM^{-\zeta(\NM)}.\]  Also, let \[\zeta(\infty){=}\lim_{\NM\rightarrow\infty}{\zeta(\NM)}.\]

To achieve $\left\|\nabla \mathcal{P}\right\|\rightarrow 0$ in \eqref{grad1}, both $\sqrt{\NM} (r-1)^2$ and $D$ should go to zero. From (\ref{Dterm}), for $D\rightarrow 0$, we need $y_{_{\NM}}=a{r^{\NM-1}}\rightarrow\infty$. Substituting $r = 1+\NM^{-\zeta(\NM)}$ and applying L'H\"{o}pital's rule, we get a necessary condition
$\zeta(\infty) \leq 1$.

Now, consider the first term in \eqref{grad1}, since $\zeta(\infty) \leq 1$ and $r=1+\NM^{-\zeta(\NM)}$, we see that the first term cannot go to zero faster than $\Theta(\NM^{-{3}/{2}})$. Since $\left\|\nabla \mathcal{P}\right\|$ is bounded by the magnitude of the first term in \eqref{grad1}, we conclude that in the class of
geometric sequences, the best possible scaling for $\left\|\nabla \mathcal{P}\right\|$ is $\Theta(\NM^{-{3}/{2}})$, which is achieved only with $r \ge 1 + \Theta(\NM^{-1})$.

\subsection{Proposed Magnitude Quantization Codebook}
In the following, we propose a specific solution for $r$ that achieves the optimal scaling. The uniform magnitude codebook corresponding to this value of $r$ is referred to as the optimal uniform (in dB) codebook in the remainder of this paper.

\begin{definition}\label{D_Yek} For any constant $0{<}c{<}\infty$, define the positive function $\mathcal{L}_c(n)$ for $n{>}1$ as the solution to the following equation:
\be\label{L_c}
\mathcal{L}_c(n)\left(1+\mathcal{L}_c(n)\right)^{n-1}=c.
\ee
Note that $\mathcal{L}_c(n)$ is a well defined function, since the left-hand side of \eqref{L_c} is monotonic in $\mathcal{L}_c(n)$ and has the range of $(0,\infty)$. Also define the function $\zeta_c(n)=-{\log \mathcal{L}_c(n)}/{\log n}$, i.e.
\be\label{zeta}
\mathcal{L}_c(n)=n^{-\zeta_c(n)}.
\ee
It can be shown that for any $0{<}c{<}\infty$,
\begin{IEEEeqnarray}{ccc}\label{limit1}
\lim_{n\rightarrow\infty} \zeta_c(n) &=& 1.
\end{IEEEeqnarray}
\end{definition}

\begin{definition}\label{uniform-mag-code} Define the \emph{optimal uniform magnitude quantization codebook} $\mbb{Y}^\star$ as
\be\label{ystar}
\mbb{Y}^\star\defined\mbb{Y}^{(g)}(r^\star),
\ee
with \[r^\star=1+\mathcal{L}_{\eta/a}(\NM),\] where $a=F^{-1}(q)$ and $\eta=\lim_{y\rightarrow\infty}{{-f(y)}/{f'(y)}}$. Also define the corresponding \emph{quantized magnitude} $\tilde{Y}^\star$ similarly as in \eqref{quant-mag}.
\end{definition}

Now the norm of the gradient with respect to $\mbb{Y}^\star$ can be bounded as follows.
\begin{theorem}\label{T_Yek} Assuming the regularity conditions, we have the following for the optimal uniform codebook $\mbb{Y}^\star$ and any $\epsilon>0$:
\be\label{theorem1}
\left\|\nabla_{\mbb{Y}^\star} \mathcal{P}\right\| < \frac{1}{2}\mu \NM^{-\left(2\zeta_{\eta/a}\left(\NM\right)-\frac{1}{2}\right)}+o\left(\NM^{-3+\epsilon}\right),
\ee
where $\mu$ is defined in Lemma \ref{L_Yek}.
\end{theorem}
\begin{proof}
See Appendix \ref{A_Yek}.
\end{proof}

Now, considering the bound in \eqref{theorem1} and noting that $\lim_{n\rightarrow\infty} \zeta_c(n) {=} 1$, we have the following for any $\epsilon>0$:
\be\label{corollary1}
\left\|\nabla_{\mbb{Y}^\star} \mathcal{P}\right\| = O\left(\NM^{-\frac{3}{2}+\epsilon}\right),
\ee
which proves that the uniform codebook $\mbb{Y}^\star$ is asymptotically optimal as $\NM{\rightarrow}\infty$ and furthermore it obtains the best scaling possible within the class of uniform (in dB) codebooks. Such optimality is not affected by the channel magnitude distribution as long as the regularity conditions in Definition \ref{reg-cond} are valid.

It should be noted that the codebook that achieves the optimal scaling is not unique\footnote{For example, the function $\mc{L}(\NM)=\NM^{-1+{\frac{1}{\sqrt{\log\NM}}}}$ can also be shown to achieve the scaling of $O(\NM^{-3/2})$.}. However, the function $\mc{L}_{c}(\NM)$ in \eqref{L_c} with $c=\eta/a$ is specifically defined so that we achieve the residual scaling of $o(\NM^{-3})$ in \eqref{theorem1}. Moreover, this definition provides a close approximation of the optimal magnitude quantization levels as shown in Fig. \ref{Fig2}. The figure compares the uniform (in dB) quantization levels in $\mbb{Y}^\star$ with the optimal quantization levels for a chi-square random variable $Y$ with the pdf
\be\label{chi}
f(y)=\frac{1}{2^{\frac{M}{2}}\Gamma(\frac{M}{2})}y^{\frac{M}{2}-1}e^{-\frac{y}{2}},
\ee
where the integer $M{>}2$ is the distribution parameter and $\Gamma(\cdot)$ is the gamma function. This is in fact the distribution of $\|\bh\|^2$ for a $M$-antenna channel vector $\bh=[h_1,h_2,\cdots,h_M]^T$ where $h_i$'s are independent and distributed according to $\mc{N}(0,1)$. The optimal quantization levels are computed by numerical minimization of \eqref{Pave} with multiple random start points.

\begin{figure}
\centering
\includegraphics[width=3.5in]{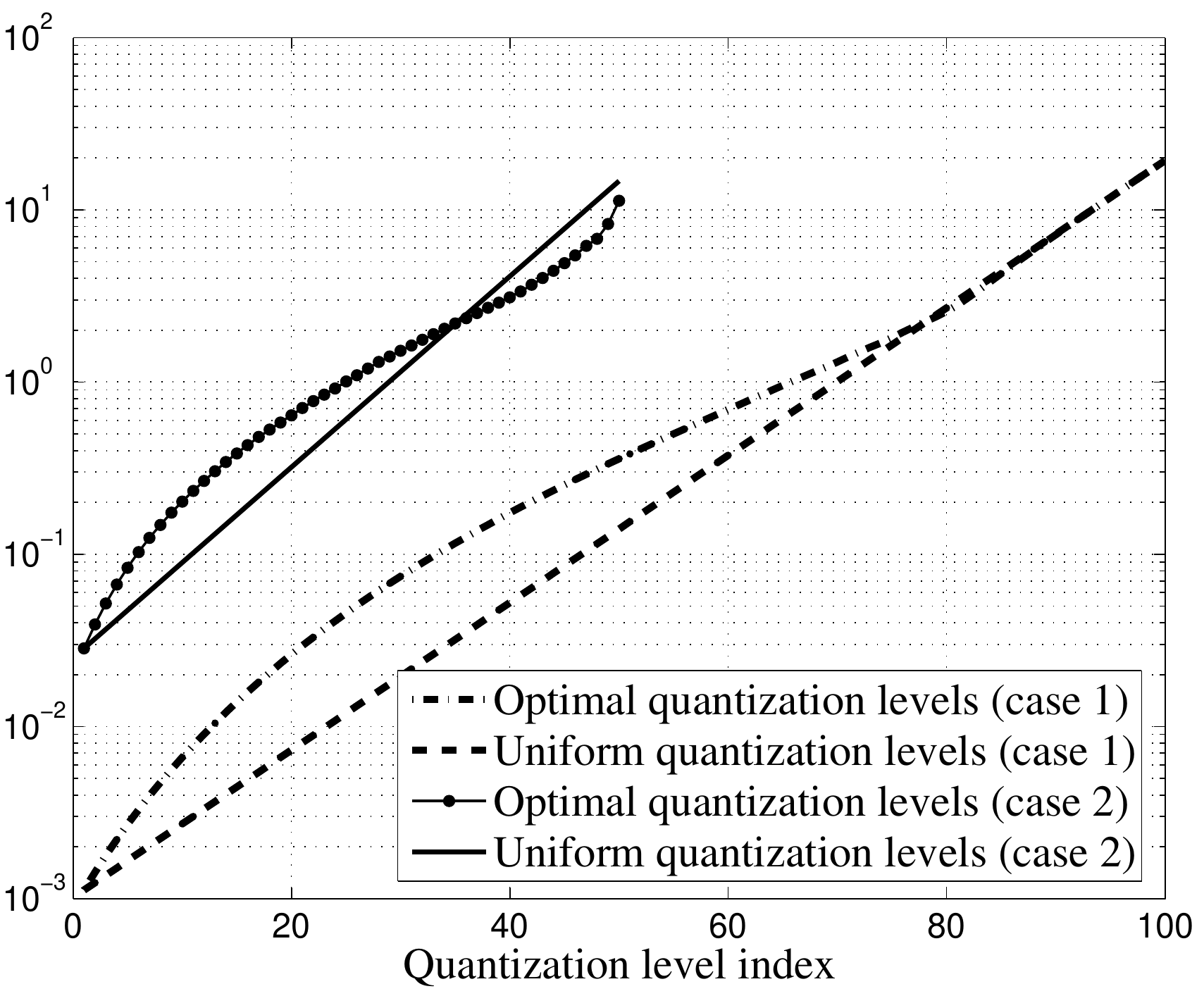}
\caption{Uniform (in dB) magnitude quantization codebook $\mbb{Y}^\star$ vs. the optimal quantization levels for the chi-square random variable $Y$; Case 1) $M{=}3$, $q{=}10^{-5}$, and $\NM{=}100$; Case 2) $M{=}4$, $q{=}10^{-4}$, and $\NM{=}50$}
\label{Fig2}
\end{figure}

In the remainder of this paper, we use the uniform codebook $\mbb{Y}^\star$ for channel magnitude quantization. In order to study the performance of this codebook in terms of the corresponding average transmission power, we present an upper bound on the average normalized transmission power $\mc{P}(\mbb{Y}^\star)=\mathbb{E}[{1}/{\tilde{Y}^\star}]$, where $\mc{P}(\cdot)$ is defined in \eqref{Pave}, and $\tilde{Y}^\star$ is the quantized magnitude for the optimal codebook $\mbb{Y}^\star$ as described in Definition \ref{uniform-mag-code}.

First, assume a perfect CSI system where the transmission power $P(h)$ is given by \eqref{CSI-power}. The normalized average transmission power is therefore given by
\be{\label{rho-CSI-siso}}\rho_{_{\md{SU,CSI}}}=\mathbb{E}\left[\frac{1}{Y}\right]=\int_{0}^{\infty}{\frac{1}{y}f(y) \textrm{d}y},\ee
where the subscript $\md{SU}$ stands for single-user.

The following theorem presents an upper bound on $\mc{P}(\mbb{Y}^\star)=\mathbb{E}[{1}/{\tilde{Y}^\star}]$.

\begin{theorem}\label{T_Do}
Assuming the regularity conditions on the channel distribution function, we have the following for $\NM{>}1$:
\be\label{powercomp}
\mathbb{E}\left[\frac{1}{\tilde{Y}^\star}\right]<\mathbb{E}\left[\frac{1}{Y}\right]\left(1+\NM^{-\zeta_{\eta/a}\left(\NM\right)}+\omega\NM^{-2\zeta_{\eta/a}\left(\NM\right)}\right),
\ee
where $\omega=\frac{\mbb{E}[Y]}{\eta^2\rho_{_{\md{SU,CSI}}}}$.
\end{theorem}
\begin{proof}
See Appendix \ref{A_Yek}.
\end{proof}


This result is used in the subsequent sections for optimizing the magnitude-direction product quantization codebooks. In Section \ref{S_Se}, we present the general formulation for the optimal design of the limited feedback system and describe its robust variant. By considering the robust formulation, Section \ref{S_Char} addresses the product channel quantization codebook optimization and presents the asymptotic bit allocation laws for the channel magnitude and direction quantization.

\section{Vector Channel Quantization: The General Form} \label{S_Se}

Consider a single-user MISO channel $\mbf{h}{\in} \mathds{R}^M$, where $M$ is the number of base station antennas. The input signal $s$ and the output signal $r$ are related as
\be\label{miso} r=\sqrt{P(\bh)}~{\bh^T\bv(\bh)} s + z, \ee
where $P(\bh)$ and $\bv(\bh)$ are the transmission power and the unit-norm transmit beamforming vector and $z\sim\mc{N}(0,1)$ is the receiver noise. The received SNR for this setup is given by $P(\mbf{h})\left|\mbf{h}^T\mbf{v(h)}\right|^2$.

Assume a limited feedback system with perfect CSI at the receiver and a feedback link with capacity $B$ bits per fading block. The receiver quantizes the channel $\bh$ and sends back the corresponding $B$-bit quantization region index. The transmitter then determines the transmission power and the beamforming vector using the quantized information.

We are interested in a system design that minimizes the average transmission power subject to a target outage probability $q$. The outage probability is defined as the probability that the received SNR drops below a target received SNR $\gamma{_{_0}}$.

Similar to the scalar channel quantization, the system design is a two-fold problem: 1) finding the optimal transmission power and beamforming vector given the channel quantization codebook; 2) optimization of the channel quantization codebook itself. In the following, we present the general form for such system design and present our approach in transforming it into a robust design problem.

We start by some definitions:
\begin{definition} \label{vec-code} By a \emph{vector channel quantization codebook} $\mc{C}$ of size $N$, we mean a partition of $\ds{R}^M$ into $N$ disjoint \emph{quantization regions} $S^{(n)}$, $1{\leq} n{\leq} N$:\[\mc{C}{=}\{S^{(1)},S^{(2)},\cdots,S^{(N)}\}.\] For every quantization codebook $\mc{C}$, we also define a \emph{quantization function} \[\mc{S}(\mbf{h}):\ds{R}^M\rightarrow\mc{C},\] which returns the quantization region that $\mbf{h}\in\ds{R}^M$ belongs to.
\end{definition}

For a given total number of feedback quantization bits $B$, the target SNR $\gamma_{_0}$, and the target outage probability $q$, the system design problem is formulated as follows:
\begin{IEEEeqnarray}{lll}\label{single-user-general-opt}
&\min\limits_{\substack{\mc{C},\\P(\sh),\\\bv(\sh)}} &\mathbb{E}_{\mbf{h}}\left[P(\sh)\right]\\
&\textmd{s.t.}&\md{prob}\left[P(\sh)\left|\mathbf{h}^{T}\mathbf{v}(\sh)\right|^2 < \gamma_{_0}\right] \leq q,\label{single-user-general-opt-const}\\
&& |\mc{C}|=N=2^B,\nonumber
\end{IEEEeqnarray}
where $N$ is the quantization codebook size and the optimization is over the codebook $\mc{C}$, the power control function $P(\sh):\mc{C} \rightarrow \ds{R}_{+}$, and the beamforming function $\bv(\sh):\mc{C}\rightarrow \mathfrak{U}_M$, where $\mathfrak{U}_M$ is the unit hypersphere in $\ds{R}^M$.

An exact solution to this problem is intractable and the numerical optimization approach is rather involved \cite{globecom}. Our approach for solving this problem is a suboptimal one where we fix an outage region of volume $q$ and design the system such that the target SNR is guaranteed for all other regions, which we refer to as  \emph{no-outage regions}. With this approach, the outage probability constraint in \eqref{single-user-general-opt-const} is replaced with SINR constraints that should be satisfied for all no-outage regions and therefore the problem in \eqref{single-user-general-opt} transforms to a robust optimization problem.

To make this rigorous, define the outage region $\mc{O}\subset\mc{C}$ such that $\md{prob}[\sh\in\mc{O}]=q$ and let the random variable $I(\bh)$ be the \emph{no-outage flag}, i.e. $I(\bh)=\mc{I}(\sh\in\mc{O}^c)$, where the logic true function $\mc{I}(\cdot)$ is $1$ if its logical argument is true and it is $0$ otherwise.

For a robust system design, we need to design the codebook, the power control function, and beamforming function in a robust manner such that the target SNR is guaranteed whenever $I(\bh)=1$:
\begin{IEEEeqnarray}{lll}\label{miso-robust-design}
&\min\limits_{\substack{\mc{C},\\P(\sh),\\\bv(\sh)}} &\mathbb{E}_{\mbf{h}}\left[P(\sh)\right]\\
&\textmd{s.t.}& \inf_{\bw\in S(\bh)}{P(\sh)\left|\mathbf{w}^{T}\mathbf{v}(\sh)\right|^2} \geq \gamma_{_0} I(\bh), \nonumber\\&&~~~~~~~~~~~~~~~~~~~~~~~~~~~~~~~~\md{for all}~ \mbf{h}\in\ds{R}^{M}.\label{miso-robust-constraint}
\end{IEEEeqnarray}
Note that by including the no-outage flag in the constraint \eqref{miso-robust-constraint}, this formulation returns $P(\sh){=}0$ if the receiver is in outage, i.e. $I(\bh)=0$. Also note that the outage constraint, $\md{prob}[I(\bh)=0]=q$, is implicit in this formulation.

It should be noted that the constraint \eqref{miso-robust-constraint} will not be feasible if the quantization region $S(\bh)$ includes the zero vector $\bw=0$ (unless $I(\bh)=0$). This means that the outage region $\mc{O}$ should at least include the quantization region in $\mc{C}$ that encompasses the origin.

For the robust design problem in \eqref{miso-robust-design}, the optimal power control function can be directly computed using the constraint \eqref{miso-robust-constraint} as follows:
\be\label{power-cont-fnc} P_{_{\md{SU}}}\defined P(\sh)=\frac{\gamma_{_0} I(\bh)}{\inf\limits_{\bw\in S(\bh)}{\left|\mathbf{w}^{T}\mathbf{v}(\sh)\right|^2}},\ee
where the subscript $\md{SU}$ stands for \emph{single-user}. The problem in \eqref{miso-robust-design} therefore simplifies to the following problem:
\begin{IEEEeqnarray}{lll}\label{simple-miso-robust-design}
&\min\limits_{\mc{C},\bv(\sh)} &\mathbb{E}_{\mbf{h}}\left[\frac{\gamma_{_0} I(\bh)}{\inf_{\bw\in S(\bh)}{\left|\mathbf{w}^{T}\mathbf{v}(\sh)\right|^2}}\right],
\end{IEEEeqnarray}
where the optimization is over the codebook $\mc{C}$ and the beamforming function $\bv(\sh)$.

The problem in \eqref{simple-miso-robust-design} cannot be solved without a specific assumption on the geometry of the quantization regions. In the next section, we further simplify this problem by assuming a magnitude-direction product structure for the channel quantization codebook.

\section{Product Quantization Codebook Design and Optimization} \label{S_Char}
In this section we present the structure of the product channel quantization codebook comprising a channel magnitude and a channel direction quantization codebook and study the optimal bit allocation between these two codebooks.

\subsection{Product Codebook Structure} \label{S_Char-A}

\subsubsection{Channel Magnitude Quantization} \label{S_Char-A-1}
The structure of the channel magnitude quantization is exactly the same as the structure described in Section \ref{S_Do}, with the difference that the magnitude variable is $Y=\|\bh\|^2$ instead of $Y=|h|^2$.

Let $\dot{\mc{C}}$ denote the set of quantization intervals for $\|\bh\|=\sqrt{Y}$:
\be\label{M-code-siso}\dot{\mc{C}}=\left\{J^{(1)},J^{(2)},\cdots,J^{(\NM)}\right\},\ee
where $J^{(n)}{=}[\sqrt{y^{(n)}},\sqrt{y^{(n+1)}})$. Here $y^{(n)}$'s are the quantization levels for $Y=\|\bh\|^2$, $\NM$ is the magnitude codebook size, and $y^{(\NM+1)}{\defined}\infty$. Similar to Section \ref{S_Do}, the first quantization level is fixed as follows:
\be\label{y1-new} y^{(1)}=F^{-1}(q),\ee where $F^{-1}(\cdot)$ is the inverse cdf of $Y=\|\bh\|^2$.

According to Theorem \ref{T_Do}, if we use the optimal uniform codebook $\mbb{Y}^\star$ for magnitude quantization, we have the following for $\NM{>}1$:
\be\label{PavescaleAgain}
\mathbb{E}\left[\frac{1}{\tilde{Y}^\star}\right]<\rho_{_{\md{SU,CSI}}} \left(1+\NM^{-\zeta_{\eta/a}\left(\NM\right)}+\omega\NM^{-2\zeta_{\eta/a}\left(\NM\right)}\right),
\ee
where $\rho_{_{\md{SU,CSI}}}$ is defined in \eqref{rho-CSI-siso} with $Y=\|\mbf{h}\|^2$.

\subsubsection{Channel Direction Quantization} \label{S_Char-A-2}
We use a $M$-dimensional Grassmannian codebook $\mbb{U}$ of size $\ND$ for direction quantization:  \be\label{grass-code-su} \mbb{U}=\left\{\bu^{(1)},\bu^{(2)},\cdots,\bu^{(\ND)}\right\},\ee where $\bu^{(n)}$ vectors are $M$-dimensional unit-norm Grassmannian codewords.

The direction quantization regions are formed by mapping each channel vector $\bh$ to a vector $\tilde{\bu}(\bh)\in\mbb{U}$ that has the smallest angle with $\bh$:
\be\label{dir-quant-vec} \tilde{\bu}(\bh)=\arg\min_{\bu\in\mbb{U}} \angle{(\bh,\bu)}.\ee
The vector $\tilde{\bu}(\bh)$ is referred to as the \emph{quantized direction} for the channel realization $\bh$. The corresponding quantization regions, according to the Gilbert-Varshamov bound argument \cite{Barg}, can be covered by the following spherical caps:
\be\label{dir-quant-regions-su} \ddot{\mc{C}}=\left\{B^{(1)},B^{(2)},\cdots,B^{(\ND)} \right\} \ee
where \[B^{(n)}=\left\{\mbf{w}\in\mathfrak{U}_M\left|\angle{(\mbf{w},\bu^{(n)})}<\phi\right.\right\}.\] Here $\mathfrak{U}_M$ is the unit hypersphere in $\ds{R}^M$, $B^{(n)}$ is the spherical cap around $\bu^{(n)}$ and \[\phi=\arcsin{\delta}\] is the angular opening of the caps, where $\delta$ is the minimum chordal distance of $\mbb{U}$.

Covering the direction quantization regions with these spherical cap regions allows for a closed form solution for \eqref{simple-miso-robust-design}. Moreover such enlargement of the regions increases the required transmission power and will result in an upper bound for the average transmission power.

In order to describe the dependence between the angular opening $\phi$ and the direction codebook size $\ND$, we need the following lemma:

\begin{lemma}\label{L_Do} For a $M$-dimensional real Grassmannian codebook of size $\ddot{N}$ and minimum chordal distance $\delta$, we have the following for large enough $\ddot{N}$:
\be\label{Hamming-su} \delta<4\lambda_M\ddot{N}^{-\frac{1}{M-1}}, \ee
where $\lambda_M=\left(\frac{\sqrt{\pi}\Gamma((M+1)/2)}{\Gamma({M/2})}\right)^{\frac{1}{M-1}}$.
\end{lemma}
\begin{proof}
The proof is based on the Hamming bound as in \cite{heath}, with the difference that, unlike the complex space, closed-form expressions do not exist for the surface area of the real spherical caps (see Appendix \ref{A_Do}).
\end{proof}

We therefore have the following for large enough $\ddot{N}$:
\be\label{phi-bound-su} \sin{\phi}<4\lambda_M\ddot{N}^{-\frac{1}{M-1}}.\ee
This bound, along with the bound in \eqref{PavescaleAgain}, is used in Section \ref{S_Char-B} for optimizing the product codebook structure.

\subsubsection{Outage Region and Product Codebook Structure} \label{S_Char-A-3}
According to the robust design argument in Section \ref{S_Se}, we have to specify and fix an outage region $\mc{O}$ such that $\md{prob}[\bh\in\mc{O}]=q$, where $q$ is the target outage probability. For this purpose, we set the outage region to be ball centered at origin:
\be\label{out-reg-su} \mc{O}=\left\{\bh\left|\|\bh\| < \sqrt{y^{(1)}}\right.\right\}, \ee
where $y^{(1)}$ is the first quantization level, which is fixed according to \eqref{y1-new}.

Now, the product channel quantization regions are described as follows:
\be\label{product-code-su} \mc{C}=(\dot{\mc{C}}\times\ddot{\mc{C}}) \cup \mc{O}, \ee
where $\dot{\mc{C}}$ includes the magnitude quantization regions in \eqref{M-code-siso} and $\ddot{\mc{C}}$ includes the direction quantization regions in \eqref{dir-quant-regions-su}. Each quantization region in \eqref{product-code-su}, except the outage ball, is schematically a sector-type region as shown in Fig. \ref{Fig3}. The variable $\tilde{Y}$ in the figure is the quantized magnitude variable as defined in \eqref{quant-mag}. Finally, noting \eqref{product-code-su}, the size of the product quantization codebook $\mc{C}$ is related to the magnitude and direction quantization codebook sizes as follows:
\be |\mc{C}|=\NM\ND+1.\ee

\begin{figure}
\centering
\includegraphics[width=1.7in]{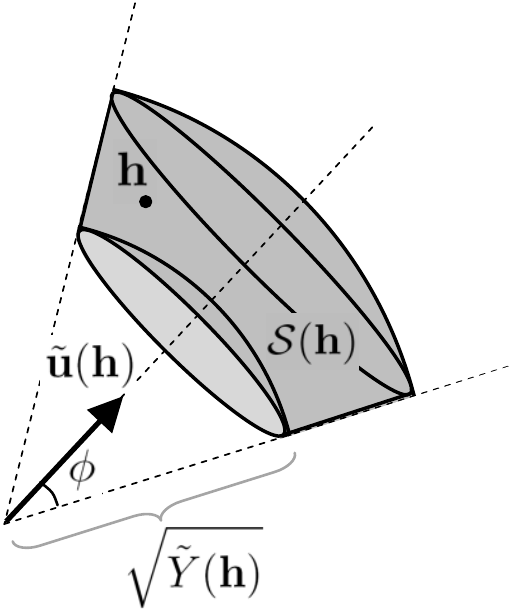}
\caption{Sector-type product quantization region $\mc{S}(\bh)$ for a channel realization $\bh$.}
\label{Fig3}
\end{figure}

Given the structure of the channel quantization codebook, the next section addresses the product codebook optimization and derives the bit allocation law between the channel magnitude and direction quantization codebooks.

\subsection{Product Codebook Optimization and the Optimal Bit Allocation Law} \label{S_Char-B}
As described in Section \ref{S_Se}, we are interested in a robust system design that minimizes the average transmission power subject to a target outage probability as formulated in \eqref{miso-robust-design}. The system design includes optimizing the power control function, the beamforming function, and the quantization codebook structure itself. The optimal power control function is given by \eqref{power-cont-fnc} in Section \ref{S_Se}. This section addresses the optimal beamforming function and the optimal product codebook structure.

First, fix the codebook $\mc{C}$. Consider the optimization of the objective function in \eqref{simple-miso-robust-design} over the beamforming function $\bv(\mc{S}(\bh))$. For a channel realization $\bh\in\mc{S}(\bh)$, the optimal transmit beamforming vector, according to \eqref{simple-miso-robust-design}, is given by
\be\label{opt-beam-vec-miso} \bv(\mc{S}(\bh))=\arg\max_{\bv\in\mathfrak{U}_M}\inf_{\bw\in S(\bh)}{\left|\mathbf{w}^{T}\mathbf{v}\right|^2}.\ee
It can be shown that the optimal beamforming vector is in fact the quantized direction itself:
\be\label{opt-vec-quant-dir} \bv(\mc{S}(\bh))=\tilde{\bu}(\bh).\ee
This is expected noting the symmetry of the quantization region around $\tilde{\bu}(\bh)$ as shown in Fig. \ref{Fig3}.

For the beamforming function $\bv(\mc{S}(\bh))=\tilde{\bu}(\bh)$, it is easy to verify that
\be\label{inf-term} \inf_{\bw\in S(\bh)}{\left|\mathbf{w}^{T}\mathbf{v}(\mc{S}(\bh))\right|^2}=\inf_{\bw\in S(\bh)}{\left|\mathbf{w}^{T}\tilde{\mathbf{u}}(\bh)\right|^2}=\tilde{Y}(\bh) \cos^2\phi.\ee
By substituting this in \eqref{power-cont-fnc} we have
\be\label{power-cont-fnc2} P_{_{\md{SU}}}=\frac{\gamma_{_0}I(\bh)}{\tilde{Y}(\bh) \cos^2\phi}.\ee
Note that $I(\bh){=}1$ for $\bh{\not\in}\mc{O}$. By taking the expectation of the both sides of \eqref{power-cont-fnc2} we obtain the following expression for the average transmission power:
\be\label{P-su} \mbb{E}[P_{_{\md{SU}}}]=\frac{\gamma_{_0}}{\cos^2\phi}\mbb{E}\left[\frac{1}{\tilde{Y}}\right],\ee
where $\tilde{Y}$ is the quantized magnitude as defined in \eqref{quant-mag}.

We are now ready to optimize the product codebook for robust system design. For this purpose, we present an upper bound for $\mbb{E}[P_{_{\md{SU}}}]$ in \eqref{P-su} and use it to formulate the product codebook optimization. The computations are asymptotic in the codebook sizes, i.e. we assume $\NM,\ND\gg 1$ throughout the optimization

The expression $\mbb{E}[{1}/{\tilde{Y}}]$ in \eqref{P-su} can be minimized by using the optimal uniform (in dB) magnitude codebook $\mbb{Y}^\star$. By setting $\tilde{Y}=\tilde{Y}^\star$ in \eqref{P-su} and using the bound in \eqref{PavescaleAgain} and also the bound on $\phi$ in \eqref{phi-bound-su}, we obtain the following upper bound for the average transmission power:
\begin{IEEEeqnarray}{lcl}\label{approx2}
\frac{\mbb{E}[P_{_{\md{SU}}}]}{P_{_{\md{SU,CSI}}}}~&<&~\frac{1+\NM^{-\zeta_{\eta/a}\left(\NM\right)}+\omega\NM^{-2\zeta_{\eta/a}\left(\NM\right)}}
{1-16\lambda_M^2\ddot{N}^{-\frac{2}{M-1}}}\nonumber\\
&\stackrel{\md{(a)}}{\approx}&~ \frac{1+\NM^{-1}}{1-16\lambda_M^2\ddot{N}^{-\frac{2}{M-1}}} \nonumber\\
&\stackrel{\md{(b)}}{\approx}&~ {1+\NM^{-1}+16\lambda_M^2\ddot{N}^{-\frac{2}{M-1}}},
\end{IEEEeqnarray}
where $P_{_{\md{SU,CSI}}}\defined \gamma_{_0}\rho_{_{\md{SU,CSI}}}$. The approximation (a) holds since according to \eqref{limit1} \[\lim_{\NM\rightarrow\infty}{\zeta_{\eta/a}(\NM)}=1.\] Also the approximation (b) holds since we assume $\NM,\NM\gg 1$.

For a total number of feedback bits $B$, we have the following constraint on the magnitude and direction codebook sizes:
\[ \NM\ND + 1 = N =  2^B,\]
where $N$ is the size of the product quantization codebook $\mc{C}$. Since the computations are asymptotic in $\NM$ and $\ND$, we approximate this constraint as
\[\NM\ND = 2^B.\]

The codebook size optimization therefore simplifies to the following problem:
\begin{IEEEeqnarray}{lll}\label{size-optim-su}
&\min_{\NM,\ND}~~& \NM^{-1}+16\lambda_M^2\ddot{N}^{-\frac{2}{M-1}}\\
&\md{s.t.}~~&  \NM \ND = 2^B. \nonumber
\end{IEEEeqnarray}
The following theorem gives the optimal channel magnitude and direction codebook sizes:

\begin{theorem}\label{T_Se}
Let $\BM\defined\log(\NM)$ and $\BD\defined\log(\ND)$ denote respectively the number of quantization bits assigned to channel magnitude and channel direction quantization. By solving \eqref{size-optim-su} the optimal number of quantization bits are given by
\begin{IEEEeqnarray}{ccc}\label{opt-bit-su}
\BM &=& \frac{2}{M+1}B-\kappa_{_{\md{SU}}}\\
\BD &=& \frac{M-1}{M+1}B +\kappa_{_{\md{SU}}},\label{opt-bit-su2}
\end{IEEEeqnarray}
where $\kappa_{_{\md{SU}}}\defined\frac{M-1}{M+1}\log{\frac{{32}\lambda_M^2}{M-1}}$.
\end{theorem}
\begin{proof}
The optimal number of bits are derived simply by applying the Lagrange multipliers method to the problem in \eqref{size-optim-su}.
\end{proof}

\begin{corollary}\label{C_Se}
For a single user system with $M$ antennas at the base station, the number of magnitude and direction quantization bits are related as follows:
\be\label{su-mag-dir-res} \BD=\frac{M-1}{2}\BM+\frac{M+1}{2}\kappa_{_{\md{SU}}},\ee
where $\kappa_{_{\md{SU}}}$ is defined in Theorem \ref{T_Se}. In the asymptotic regime of $B\rightarrow\infty$, the optimal number of magnitude and direction quantization bits are therefore related as \be\label{asym-res} \BD=\frac{M-1}{2}\BM.\ee
\end{corollary}

It should be noted that the codebook optimization problem, as described in this section, can be viewed as a classical constrained quantization codebook design problem if one considers the average transmission power in \eqref{miso-robust-design} as the quantization distortion measure. The relative bit allocation law in \eqref{asym-res} can specifically compared with the results in \cite{Hamkins}, which studies the optimal shape (channel direction) and gain (channel magnitude) quantization with the mean squared error (MSE) distortion measure. The authors show that, with the MSE measure, the asymptotic number of shape-gain quantization bits are related as \be\label{MSE-bit}\BD=(M-1)\BM.\nonumber\ee Comparing this with \eqref{asym-res} implies that the number of direction quantization bits for a product codebook optimized for the MSE measure differs by a factor of two.

Theorem \ref{T_Se} is the final part of the product codebook optimization for the single-user system. By combining this structure with the beamforming and power control functions in \eqref{opt-vec-quant-dir} and \eqref{power-cont-fnc2}, we achieve a robust system design for a limited-feedback single-user MISO system.

Finally, in order to derive the performance scaling of the designed system  with the feedback rate $B$, we substitute the optimal codebook sizes of Theorem \ref{T_Se} into \eqref{approx2} and compute the average transmission power.

\begin{theorem}\label{T_Char} For a single-user system with $M$ antennas at the base station and a total feedback rate of $B$ bits per fading block, we have the following asymptotic upper bound as $B\rightarrow\infty$:
\be\label{su-pow-scaling} \mbb{E}[P_{_{\md{SU}}}]<P_{_{\md{SU,CSI}}}\left(1+\sigma_{_{\md{SU}}}2^{-\frac{2B}{M+1}}\right),\ee
where $P_{_{\md{SU,CSI}}}=\gamma_{_0}\rho_{_{\md{SU,CSI}}}$ and \[\sigma_{_{\md{SU}}}=\frac{16(M+1)}{M-1}\left(\frac{\sqrt{\pi}(M-1)\Gamma((M+1)/2)}{{32}{\Gamma(M/2)}}\right)^{{2}/{(M+1)}}.\]
\end{theorem}
\begin{proof}
The proof uses basic calculations, which are omitted due to the space limits.
\end{proof}

Thus, as the feedback rate $B$ increases, the performance of the designed system approaches the performance of the perfect CSI system as $2^{-\frac{2B}{M+1}}$.

\subsection{Complex Space Channels}\label{complex_channels}

As mentioned earlier, the results in Section \ref{S_Char-B} are based on the assumption of real space channels. The same approach however can be used to derive the bit allocation laws for complex space channels. The only main difference is in the use of Hamming bound in \eqref{Hamming-su}. For the complex space channels, the Hamming bound is as follows \cite{heath}:
\be\label{complex-hamming} \delta<2\ddot{N}^{-\frac{1}{2(M-1)}}.\ee
By using this bound instead of \eqref{Hamming-su}, the remainder of the analysis can be applied in a similar fashion to derive the optimal magnitude-direction bit allocation law. By doing so, one can show that for complex channels and in the asymptotic regime of $B\rightarrow\infty$, the number of magnitude and direction quantization bits are related as $\ddot{B}{=}(M{-}1)\dot{B}$ instead of the law in \eqref{asym-res} for real channels \cite{queens}. Furthermore, for the complex space, the the system performance scales with the feedback rate as $2^{-\frac{B}{M}}$ instead of $2^{-\frac{2B}{M+1}}$ for the real space.

\section{Numerical Results} \label{numerical_results}

In the process of deriving optimal bit allocation laws in Theorem \ref{T_Se}, we used upper bounds in \eqref{theorem1} and \eqref{phi-bound-su} and the approximations in \eqref{approx2} to find an analytically tractable approximation of the original objective function in \eqref{P-su}. In order to examine the accuracy of the magnitude and direction bit allocation laws in Theorem \ref{T_Se}, this section compares these closed-form solutions with the corresponding numerically optimized bit allocations.

In order to find the optimal bit allocations, we need to numerically compute and minimize the average transmission power in \eqref{P-su} in terms of the magnitude and direction codebook sizes $\dot{N}$ and $\ddot{N}$. The minimization is subject to the constraint $\dot{N}\ddot{N}{=}2^B$ or equivalently $\dot{B}{+}\ddot{B}{=}B$, where $\dot{B}$ and $\ddot{B}$ are the number of bits assigned to magnitude and direction quantization and $B$ is the total number of feedback bits.

The expression for the average transmission power in \eqref{P-su} comprises two terms that are controlled independently by the magnitude and direction quantization codebooks. The term $\mbb{E}[{1}/{\tilde{Y}}]$ in \eqref{P-su} is controlled by the channel magnitude quantizer and depends on the magnitude quantization codebook size $\dot{N}$. The variable $\phi$ on the other hand is controlled by the direction quantizer and depends on the direction codebook size $\ddot{N}$.

For the channel magnitude quantization, we use the uniform (in dB) codebook $\mbb{Y}^\star$ defined in Definition \ref{uniform-mag-code} and numerically compute the term $\mbb{E}[{1}/{\tilde{Y}}]$ using \eqref{Pave}. In our numerical results, we use the chi-square distribution in \eqref{chi} with $M=5$ as the channel magnitude distribution. For the direction quantization, on the other hand, we use real-space Grassmannian codebooks and rely on the available numerical tables that give the angular opening $\phi$ as a function of the codebook size $\ddot{N}$ \cite{grass_code}. The values of $\sin{\phi}$ for channel space dimension $M{=}5$ and codebook sizes up to $\ddot{N}{=}100$ are shown in Fig. \ref{Fig6}.

\begin{figure}
\centering
\includegraphics[width=3.53in]{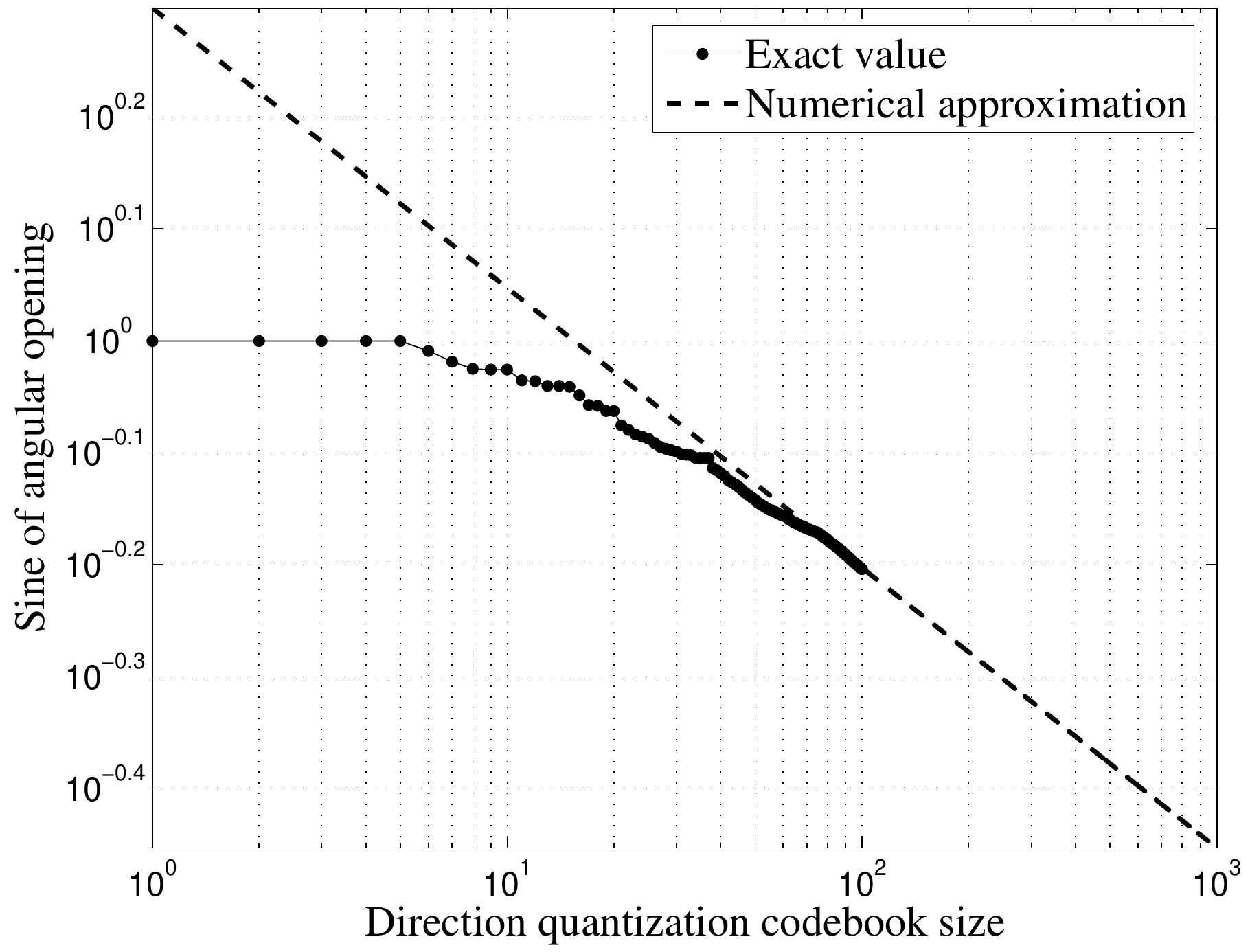}
\caption{Sine of angular opening $\phi$ as a function of direction codebook size $\ddot{N}$ for Grassmannian codebooks in $\mathds{R}^5$.}
\label{Fig6}
\end{figure}

Since the numerical tables for Grassmannian codebooks are only available for moderate values of $\ddot{N}$, we need to derive an extrapolation of the available data for larger values of $\ddot{N}$. For this purpose we fit a line in the least-squares sense to the available data $(\ddot{N},\sin{\phi})$ in the logarithmic scale. Fig. \ref{Fig6} shows a line fitting that uses the data points with $80\leq\ND\leq 100$ as input. The slope of this line is forced to be $-\frac{1}{M-1}$ in order to match the Hamming upper bound in \eqref{phi-bound-su}\footnote{According to the Gilbert-Varshamov bound \cite{Barg,heath}, we have \[\sin{\phi}>\alpha\ddot{N}^{-{1}/{(M-1)}},\] for some constant $\alpha>0$. By combining this bound with the Hamming bound in \eqref{phi-bound-su}, it is evident that $\sin{\phi}$ should asymptotically scale as $\ddot{N}^{-{1}/{(M-1)}}$ with the codebook size $\ddot{N}$. The slope of the fitted line in logarithmic scale should therefore be equal to $-\frac{1}{M-1}$.}. This leads to the following approximation of $\sin{\phi}$ for codebook sizes $\ddot{N}>100$:
\be\label{phi_approx}
\sin{\phi}\approx {1.9833} {\ddot{N}^{-\frac{1}{M-1}}},
\ee
where $M=5$ is the channel space dimension or equivalently the number of base station antennas. In numerical computation of the objective function in \eqref{P-su}, we rely on the approximation in \eqref{phi_approx} whenever $\ND>100$.

To find the optimal bit allocation for a given number of feedback bits $B$, we numerically compute and compare the objective function \eqref{P-su} for different integer pairs $(\dot{B},\ddot{B})$ that satisfy $\dot{B}{+}\ddot{B}{=}B$ and choose the best pair. The resulting bit allocation is shown in Fig. \ref{Fig7}, where the target outage probability is set to $q={10}^{-4}$. The figure also compares the results with the values of $\dot{B}$ and $\ddot{B}$ in \eqref{opt-bit-su} and \eqref{opt-bit-su2}, which are rounded to the closest integer numbers. As the figure shows the analytically optimized bit allocations in Theorem \ref{T_Se} are at most two bits away from the numerically optimized ones. Furthermore, as the number of feedback bits $B$ increases, the asymptotic result of $\ddot{B}=\frac{M-1}{2}\dot{B}$ becomes more accurate.

\begin{figure}
\centering
\includegraphics[width=3.5in]{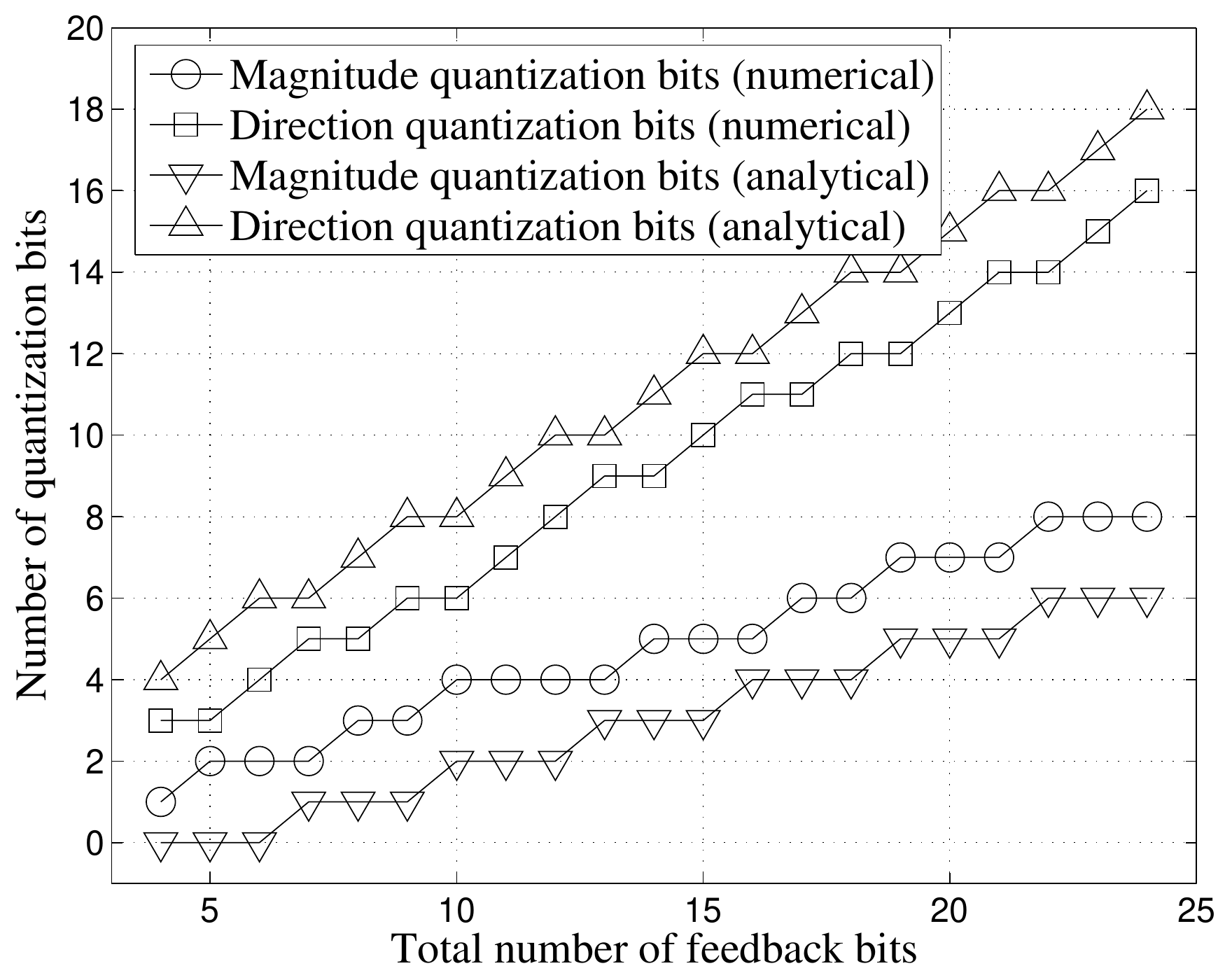}
\caption{Magnitude and direction bit allocations for $M=5$ antennas and target outage probability $q={10}^{-4}$.}
\label{Fig7}
\end{figure}

In order the evaluate the system performance under these bit allocations, we define a distortion measure as follows:
\be\label{dist_meas}
\mathcal{D}(B)=\frac{\mbb{E}[P_{_{\md{SU}}}]-P_{_{\md{SU,CSI}}}}{P_{_{\md{SU,CSI}}}},
\ee
where $\mbb{E}[P_{_{\md{SU}}}]$ is the average transmission power with limited feedback as defined in \eqref{P-su} and $P_{_{\md{SU,CSI}}}=\gamma_{_0}\rho_{_{\md{SU,CSI}}}$ is the average transmission power with perfect CSI as defined in Theorem \ref{T_Char}, and $\gamma_{_0}$ is the target SNR. For the channel magnitude with chi-square distribution in \eqref{chi}, one can show that $\rho_{_{\md{SU,CSI}}}=\frac{1}{M-2}$ and therefore $P_{_{\md{SU,CSI}}}=\frac{\gamma_{_0}}{M-2}$.

The distortion measure $\mathcal{D}(B)$ clearly depends on the total number of feedback bits $B$ and also on how these bits are split between the channel magnitude and direction quantization codebooks. It should be noted however that the distortion measure does not depend on the target SNR, since $\gamma_{_0}$ appears as a multiplicative factor both in the enumerator and denominator. As a result, the optimal magnitude-direction bit allocation, which minimizes the distortion or equivalently $\mbb{E}[P_{_{\md{SU}}}]$, does not depend on the target SNR as it is verified by Theorem \ref{T_Se}.

Fig. \ref{Fig8} compares the distortion measure when one uses the numerically optimized bit allocations versus the case where bit allocations in \eqref{opt-bit-su} and \eqref{opt-bit-su2} are used. As it is shown in the figure, the two bit allocations show a close performance, which verifies the accuracy of the closed-form bit allocation laws in Theorem \ref{T_Se}. Fig. \ref{Fig8} also shows the distortion measure applied to the upper bound in Theorem \ref{T_Char}, which leads to the following asymptotic upper bound:
\be\label{dist_bound}
\mc{D}(B)<\sigma_{_{\md{SU}}}2^{-\frac{2B}{M+1}}.
\ee

\begin{figure}
\centering
\includegraphics[width=3.6in]{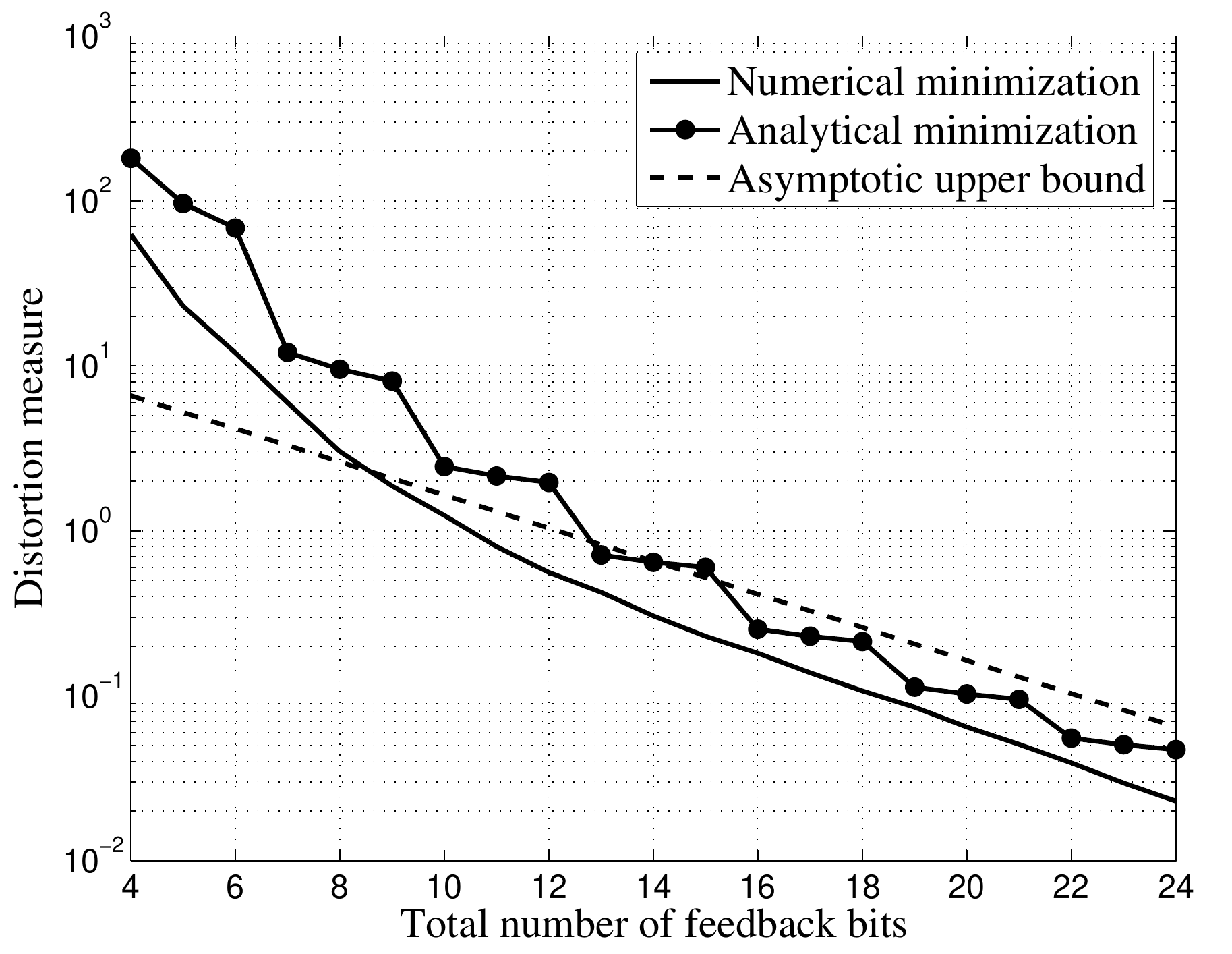}
\caption{Distortion measure $\mathcal{D}(B)$ for $M=5$ antennas and target outage probability $q={10}^{-4}$. }
\label{Fig8}
\end{figure}

Although the quantization codebook design objective used in our work, i.e. minimization of the average transmission power subject to a target outage probability, is quite different from those used in standard quantization theory, the distortion measure upper bound in \eqref{dist_bound} still resembles the asymptotic distortion bounds derived in the literature of high resolution quantization theory \cite{gersho,zheng}. The authors of \cite{zheng} specifically develop a rather general framework for analysis of high resolution vector quantization with locally quadratic distortion functions and show that the average distortion function scales asymptotically with the number of quantization bits $B$ as $2^{-\frac{2B}{k}}$, where $k$ is the \emph{degrees of freedom} of the vector under quantization. One would therefore naturally expect a denominator of $M$ instead of $M{+}1$ in the exponent of the upper bound in \eqref{dist_bound}, since the channel vector is $M$-dimensional. This difference, we believe, is due to the difference between the distortion measure used in this paper with the standard quadratic (e.g. MSE) distortion measures used in the literature. Nonetheless, it is still possible to present an intuitive justification, although inevitably inaccurate one, for the difference in the asymptotic scaling. We describe this justification next.

We first note that for a $M$-dimensional channel vector in real space, the channel magnitude is a one-dimensional variable. The channel direction, on the other hand, is $M{-}1$ dimensional due to its unit-norm constraint. Now, according to the asymptotic bit allocation law in \eqref{asym-res}, every bit assigned to channel magnitude quantization translates to two bits for channel direction quantization. It is therefore possible to say that in the process of dividing the available number of bits $B$ among the available degrees of freedom, the degrees of freedom in the channel magnitude space counts twice the degrees of freedom in the channel direction space. Hence, we can assume an equivalent degrees of freedom and calculate it as $(M-1) + 2 = M+1$.

A similar interpretation can also be applied to channel quantization in complex space. In this case, the channel is $2M$ dimensional, the channel magnitude is one-dimensional, and the channel direction is $2M-1$ dimensional. Nevertheless, as far as the power control and beamforming problem in this paper is concerned, any channel vector $\mathbf{h}$ is equivalent to all the channels in the form of $e^{j\theta}\mathbf{h}$, where $\theta$ is an arbitrary real number. This means that the channel direction space has effectively $2M-2$ degrees of freedom. Now by following the intuition described for real space channels, we can assume an equivalent degrees of freedom for channel quantization and calculate it as $(2M-2) + 2 = 2M$. This would lead to an asymptotic scaling of $2^{-\frac{2B}{2M}}=2^{-\frac{B}{M}}$, which coincides with the scaling mentioned in Section \ref{complex_channels} for complex channels.

\section{Conclusions} \label{S_Panj}

This paper studies the design of single-user MISO system with limited CSI at the base station. The problem is formulated as the minimization of the average base station transmission power subject to the outage probability constraint at the user side. We show that the asymptotically optimal channel magnitude quantization codebook is uniform in dB scale. This result does not depend on the channel magnitude distribution function as long as some regularity conditions are satisfied. Combining the uniform in dB magnitude codebook with a spatially uniform channel direction quantization codebook, we form a product channel quantization codebook and optimize the quantization bit allocation as $B\rightarrow\infty$, where $B$ is the total number of feedback quantization bits. It is shown that for channels in real space, as $B\rightarrow\infty$, the optimal number of direction quantization bits is $\frac{M-1}{2}$ times the number of magnitude quantization bits, where $M$ is the number of BS antennas. We also show that, as $B$ increases, the performance of the designed system approaches the performance of perfect CSI system as $2^{-\frac{2B}{M+1}}$. For channels in complex space, on the other hand, the number of magnitude and direction quantization bits are related by a factor of $(M{-}1)$ and the system performance scales as $2^{-\frac{B}{M}}$ as $B\rightarrow\infty$.

\section*{Acknowledgments}
The authors wish to acknowledge the anonymous reviewers for many helpful comments, especially for pointing out an intuitive connection between the asymptotic performance scaling result in this paper and that of \cite{zheng}.

\appendices

\section{}\label{A_Yek}
For notation convenience, the notations $y^{(n)}$ and $Q^{(n)}$ in Section \ref{S_Do} are replaced with $y_n$ and $Q_n$.

\subsection{Proof of Lemma \ref{L_Yek}}
\begin{proof}
We first note that $y_1{=}a{=}F^{-1}(q)$ is fixed according to \eqref{y1}. By taking the derivative of \eqref{Pave} with respect to $y_n$ for $2 {\leq} n {\leq}\NM$, we have
\be\label{1} \frac{\partial \mc{P}}{\partial y_n}=f(y_n)\left[\frac{1}{y_{n-1}}-\frac{1}{y_n}\right]-\frac{Q_n}{y_n^2}.\ee

For $2{\leq} n{\leq} \NM$, we have the following by adding and subtracting a fixed term to \eqref{1}:
\begin{IEEEeqnarray}{lll}\label{2}  \frac{\partial \mc{P}}{\partial y_n}&=&f(y_n)\left[\frac{1}{y_{n-1}}-\frac{1}{y_n}\right]- f(y_n)\left[\frac{y_{n+1}-y_n}{y_n^2}\right]\\
&+&f(y_n)\left[\frac{y_{n+1}-y_n}{y_n^2}\right]-\frac{Q_n}{y_n^2}.\nonumber
\end{IEEEeqnarray}
The first two terms in \eqref{2} cancel out since $y_n=ar^{n-1}$ form a geometric sequence. By using the definition of $Q_n$, we therefore have the following for $2 {\leq} n {\leq}\NM-1$:
\be\label{4}
\frac{\partial \mc{P}}{\partial y_n}=\frac{1}{y_n^2}\left[f(y_n)\ ({y_{n+1}-y_n})-\ \int_{y_n}^{y_{n+1}}{f(y)\ \du y}\right].
\ee

By applying Taylor's expansion to the function $g(t)=\int_{0}^{t}{f(y)\ \du y}$ for $y_n{\leq} t{\leq} y_{n+1}$, we have
\be\label{5}
\left|\int_{y_n}^{y_{n+1}}{f(y) \ \du y}\ - \ f(y_n)\ ({y_{n+1}-y_n}) \right|<\frac{1}{2}\mu ({y_{n+1}-y_n})^2,
\ee
where $\mu$ is an upper bound on $g''(t)=f'(t)$ for $y_n{\leq} t{\leq} y_{n+1}$.

By combining \eqref{4} and \eqref{5}, we therefore have the following for $2 {\leq} n {\leq}\NM-1$:
\be\label{6}
\left|\frac{\partial \mc{P}}{\partial y_n}\right|<{\frac{1}{2}\mu ({y_{n+1}-y_n})^2}/{y_n^2}=\frac{1}{2}\mu(r-1)^2.
\ee

Now for the last quantization level $y_n$, $n=\NM$, we have
\begin{IEEEeqnarray}{lll}\label{7}
\left|\frac{\partial \mc{P}}{\partial y_n}\right|&=&f(y_n)\left[\frac{1}{y_{n-1}}-\frac{1}{y_n}\right]-\frac{Q_n}{y_n^2}\nonumber\\
&=&\left|\frac{1}{y_n^2}\left(y_n f(y_n)\left[\frac{y_n}{y_{n-1}}-1\right]-Q_n\right)\right|_{n{=}\NM}\nonumber\\
&=&\left|\frac{1}{y_n^2}\left((r{-}1)y_nf(y_n){-}\int_{y_n}^{\infty}{f(t) \mathrm{d} t}\right)\right|_{n{=}\NM}{=}D.~~~~~~~
\end{IEEEeqnarray}

By combining \eqref{6} and \eqref{7} we therefore have:
\begin{IEEEeqnarray}{lll}\label{8}
\left\|\nabla_{\mbb{Y}^{(g)}(r)} \mathcal{P}\right\|&=&\left(\sum_{n=2}^{\NM-1}{\left|\frac{\partial \mc{P}}{\partial y_{n} }\right|^2}+\left|\frac{\partial \mc{P}}{\partial y_{\NM}}\right|^2\right)^{\frac{1}{2}}\nonumber\\
&\leq& \left(\sum_{n=2}^{\NM-1}{\left|\frac{\partial \mc{P}}{\partial y_n}\right|^2}\right)^{\frac{1}{2}} +D\nonumber\\
&<& \frac{1}{2}\mu\sqrt{\NM}\ {(r-1)^2}+D.
\end{IEEEeqnarray}

\end{proof}

\subsection{Proof of Theorem \ref{T_Yek}}
\begin{proof}
According to \eqref{ystar}, the geometric sequence parameter for $\mbb{Y}^\star$ is given by
\be\label{T1-0} r^\star=1+\mathcal{L}_{\eta/a}(\NM)=1+\NM^{-\zeta_{\eta/a}\left(\NM\right)},\ee
where \[\eta=\lim_{y\rightarrow\infty}{\frac{-f(y)}{f'(y)}}.\] By substituting the value of $r^\star$ in the upper bound of Lemma \ref{L_Yek} in \eqref{8}, we have
\be\label{T1-1}
\left\|\nabla_{\mbb{Y}^\star} \mathcal{P}\right\| < \frac{1}{2}\mu \NM^{-\left(2\zeta_{\eta/a}\left(\NM\right)-\frac{1}{2}\right)}+D,
\ee
where \be\label{T1-2} D=\left|\frac{1}{y_n^2}\left((r{-}1)y_nf(y_n){-}\int_{y_n}^{\infty}{f(t) \mathrm{d} t}\right)\right|_{n{=}\NM}.\ee

Noting the definition of $r^\star$ in \eqref{T1-0}, the definition of the function $\mc{L}_c(n)$ in \eqref{L_c}, and also that $y_{\NM}=ar^{\NM-1}$, we have
\be\label{T1-3}(r^\star-1)y_{\NM}=\mc{L}_{\eta/a}(\NM) \ a\left(1+\mc{L}_{\eta/a}(\NM)\right)^{\NM-1}=a\cdot \eta/a=\eta,  \ee
therefore
\be\label{T1-4}
y_{\NM}=\frac{\eta}{r^\star-1}=\eta \NM^{\zeta_{\eta/a}\left(\NM\right)}.
\ee
Note that according to \eqref{limit1}, $\lim_{\NM\rightarrow\infty}\zeta_{\eta/a}(\NM)=1$, and therefore $y_{\NM}\rightarrow\infty$ as $\NM\rightarrow\infty$.

In order to bound $D$ in \eqref{T1-1}, we substitute \eqref{T1-3} in \eqref{T1-2}:
\begin{IEEEeqnarray}{lll}\label{T1-5}
D&=&\frac{f(y_{\NM})}{y_{\NM}^2}\left|\eta {-}\frac{1}{f(y_{\NM})}\int_{y_{\NM}}^{\infty}{f(t) \mathrm{d} t}\right|\nonumber\\
&=&\frac{y_{\NM}f(y_{\NM})}{y_{\NM}^3}\left|\eta {-}\frac{1}{f(y_{\NM})}\int_{y_{\NM}}^{\infty}{f(t) \mathrm{d} t}\right|.
\end{IEEEeqnarray}

Now noting that $y_{\NM}\rightarrow\infty$ as $\NM\rightarrow\infty$, we have
\be\label{T1-6} \lim_{\NM\rightarrow\infty}\frac{1}{f(y_{\NM})}\int_{y_{\NM}}^{\infty}{f(t) \mathrm{d} t} =\lim_{\NM\rightarrow\infty} \frac{-f(y_{\NM})}{f'(y_{\NM})}=\eta. \ee
Moreover, according third regularity condition in Definition \ref{reg-cond} in Section \ref{S_Do}, we have
\be\label{T1-7} \lim_{\NM\rightarrow\infty} y_{\NM}f(y_{\NM}) =0. \ee

Combining \eqref{T1-6} and \eqref{T1-7} with \eqref{T1-5}, we have the following as $\NM\rightarrow\infty$:
\be\label{T1-8} D=o\left({1}/{y_{\NM}^3}\right)=o\left(\NM^{-3\zeta_{\eta/a}\left(\NM\right)}\right).    \ee
And since $\lim_{\NM\rightarrow\infty}\zeta_{\eta/a}(\NM)=1$, we have $D=o(\NM^{-3+\epsilon})$, for any arbitrary $\epsilon>0$. By substituting this in \eqref{T1-1}, the proof is complete.

\end{proof}

\subsection{Proof of Theorem \ref{T_Do}}

\begin{proof}
\begin{figure}
\centering
\includegraphics[width=2in]{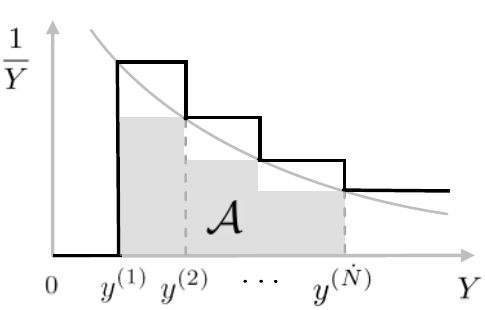}
\caption{Proof of Theorem \ref{T_Do}. }
\label{Fig5}
\end{figure}

Consider the shaded area $\mc{A}$ in Fig. \ref{Fig5}. We clearly have
\be\label{T2-1}
\mc{A}<\int_{0}^{\infty}\frac{1}{y} f(y)\ \du y =\mbb{E}[\frac{1}{Y}]=\rho_{_{\md{SU,CSI}}}.
\ee

On the other hand, according to Fig. \ref{Fig5},
\be\label{T2-2}
\mc{A}=\sum_{n=1}^{\NM-1}{\frac{Q_n}{y_{n+1}}}=\frac{1}{r^\star}\sum_{n=1}^{\NM-1}{\frac{Q_n}{y_n}}=\frac{1}{r^\star}\left(\mc{P}(\mbb{Y}^\star)-\frac{Q_{\NM}}{y_{\NM}}\right).
\ee
By combining \eqref{T2-1} and \eqref{T2-2}, we have
\be\label{T2-3}
\mc{P}(\mbb{Y}^\star)< r^\star\rho_{_{\md{SU,CSI}}}+\frac{Q_{\NM}}{y_{\NM}}= \left(1+\NM^{-\zeta_{\eta/a}\left(\NM\right)}\right)\rho_{_{\md{SU,CSI}}}+\frac{Q_{\NM}}{y_{\NM}},
\ee
where we have used the definition of $r^\star$ in \eqref{T1-0}.

Now from Markov's inequality,
\be\label{T2-4} Q_{\NM}=\md{prob}[Y>y_{\NM}] < \frac{\mbb{E}[Y]}{y_{\NM}},\ee
and therefore by using \eqref{T1-4}, we have
\be\label{T2-5} \frac{Q_{\NM}}{y_{\NM}}<\frac{\mbb{E}[Y]}{y_{\NM}^2}=\frac{\mbb{E}[Y]}{\eta^2} \NM^{-2\zeta_{\eta/a}\left(\NM\right)}. \ee

By combining \eqref{T2-3} and \eqref{T2-5}, we have
\[\frac{\mc{P}(\mbb{Y}^\star)}{\rho_{_{\md{SU,CSI}}}}<1+\NM^{-\zeta_{\eta/a}\left(\NM\right)}+\omega\NM^{-2\zeta_{\eta/a}\left(\NM\right)},\]
where $\omega=\frac{\mbb{E}[Y]}{\eta^2\rho_{_{\md{SU,CSI}}}}$.
\end{proof}

\section{Proof of Lemma \ref{L_Do}}\label{A_Do}
Let $\mathfrak{U}_M$ denote the unit hypersphere in $\ds{R}^M$. Also let $\mbb{U}$ denote a Grassmannian codebook of size $\ND$ and define $\delta$ as the \emph{minimum chordal distance} of this codebook, as defined in \cite{heath}. Consider an arbitrary unit vector $\bu$ and define an spherical cap around it as \[\mc{B}_{\psi}=\left\{\mbf{w}\in\mathfrak{U}_M\left|\angle{(\mbf{w},\bu)}<\psi\right.\right\}\] where \[\psi=\arcsin{\delta/2}.\] By applying the Hamming-bound argument, as in \cite{heath}, an upper bound on the Grassmannian codebook size can be achieved as follows:
\be\label{hamming} \ND <\frac{\textsf{A}(\mathfrak{U}_M) }{\textsf{A}(\mc{B}_\psi)}, \ee
where $\textsf{A}(\cdot)$ denotes the surface area of its argument.

The area of the unit hypersphere in $\ds{R}^M$ is given by
\be\label{unit-sph-area} \textsf{A}(\mathfrak{U}_M)=M C_M,\ee
where $C_M=\frac{\pi^{M/2}}{\Gamma({M/2} +1)}$, and $\Gamma(\cdot)$ is the gamma function. The area of the spherical cap on the other hand is given by
\be\label{cap-area} \textsf{A}(\mc{B}_\psi) = 2(M-1)C_{M-1}\int_{0}^{\psi}{\sin^{M-2} \varphi \ \du \varphi}.\ee

Since the minimum chordal distance of the codebook tends to zero as codebook size tends to infinity, we have $\lim_{\ND\rightarrow\infty}{\psi}=0$. Therefore, for $0\leq \varphi\leq \psi$ and any $0<\epsilon<1$, we have the following for large enough $\ND$:
\be\label{sin-bound} \sin^{M-2} \varphi > (1-\epsilon) \varphi^{M-2}. \ee
By using this inequality in \eqref{cap-area}, we have
\be\label{cap-bound} \textsf{A}(\mc{B}_\psi) > 2(1-\epsilon)C_{M-1} \psi^{M-1}. \ee

By combining \eqref{hamming}, \eqref{unit-sph-area}, and \eqref{cap-bound}, we have  \be\label{hamming-proof} \ND<\frac{MC_M}{2(1-\epsilon)C_{M-1} \psi^{M-1}},\ee
and therefore by using $\psi=\arcsin{\delta/2}$, we achieve
\begin{IEEEeqnarray}{lll}\label{hamming-end} ((1-\epsilon){\delta}/{2})^{M-1}&<&(1-\epsilon)(\delta/2)^{M-1}=(1-\epsilon)\sin^{M-1}\psi\nonumber\\
&<& (1-\epsilon)\psi^{M-1} < \frac{MC_M}{2C_{M-1}}\ND^{-1}. \end{IEEEeqnarray}
We therefore have the following for large enough $\ND$ and any $0<\epsilon<1$:
\be\label{last} \delta < \frac{2}{1-\epsilon}\lambda_M \ND^{-\frac{1}{M-1}},\ee
where \[\lambda_M=\left(\frac{MC_M}{2C_{M-1}}\right)^{\frac{1}{M-1}}=
\left(\frac{\sqrt{\pi}\Gamma((M+1)/2)}{\Gamma({M/2})}\right)^{\frac{1}{M-1}}.\]
By setting $\epsilon=\frac{1}{2}$ in \eqref{last}, we achieve the bound in Lemma \ref{L_Do} and the proof is complete.

\end{document}